\documentclass[journal,12pt,onecolumn, draft]{IEEEtran}
\usepackage[utf8]{inputenc}
\usepackage[OT1]{fontenc}
\usepackage{amsfonts, amsmath, amsthm, amssymb}
\usepackage{graphicx}
\usepackage{listings}
\usepackage[margin=1in]{geometry}
\usepackage{xcolor}
\usepackage{hyperref}
\usepackage{xcolor}
\usepackage{multirow}
\usepackage{tabularx,booktabs}
\usepackage{lipsum}
\usepackage{diagbox}
\usepackage{tablefootnote}

\newtheorem{definition}{Definition}
\newtheorem{theorem}{Theorem}
\newtheorem{example}{Example}
\newtheorem{corollary}{Corollary}
\newtheorem{lemma}{Lemma}
\newtheorem{remark}{Remark}
\newtheorem*{construction_classic}{Construction~I}
\newtheorem*{construction_Phi_classic}{Construction~I'}
\newtheorem*{construction_Blackburn}{Construction~I.A}
\newtheorem*{construction_Bilotta}{Construction~II}
\newtheorem*{construction_Phi_Bilotta}{Construction~II'}

\providecommand{\pre}{\mathrm{Pre}}
\providecommand{\suf}{\mathrm{Suf}}
\newcolumntype{C}{>{\centering\arraybackslash}X} 
\title{$Q$-ary non-overlapping codes : a generating function approach\thanks{The authors were partially supported by National Key Research and Development Program of China (Grant No. 2018YFA0902600) and the Shenzhen fundamental research programs (Grant No. 20200925154814002).}}

\begin{document}

\author{Geyang~Wang, \and Qi~Wang
\thanks{
G. Wang is with the Department of Computer Science and Engineering, Southern University of Science and Technology, Shenzhen 518055, China (email: 11930618@mail.sustech.edu.cn).
}

\thanks{
Q. Wang is with the Department of Computer Science and Engineering, and is also with National Center for Applied Mathematics (Shenzhen), Southern University of Science and Technology, Shenzhen 518055, China (email: wangqi@sustech.edu.cn).
}

}
\maketitle

\begin{abstract}
Non-overlapping codes are a set of codewords in $\bigcup_{n \ge 2} \mathbb{Z}_q^n$, where $\mathbb{Z}_q = \{0,1,\dots,q-1\}$, such that, the prefix of each codeword is not a suffix of any codeword in the set, including itself; and for variable-length codes, a codeword does not contain any other codeword as a subword. In this paper, we investigate a generic method to generalize binary codes to $q$-ary for $q > 2$, and analyze this generalization on the two constructions given by
Levenshtein (also by Gilbert; Chee, Kiah, Purkayastha, and Wang) and Bilotta, respectively. The generalization on the former construction gives large non-expandable fixed-length non-overlapping codes whose size can be explicitly determined; the generalization on the later construction is the first attempt to generate $q$-ary variable-length non-overlapping codes. More importantly, this generic method allows us to utilize the generating function approach to analyze the
cardinality of the underlying $q$-ary non-overlapping codes. The generating function approach not only enables us to derive new results, e.g., recurrence relations on their cardinalities, new combinatorial interpretations for the constructions, and the limit superior of their cardinalities for some special cases, but also greatly simplifies the arguments for these results. Furthermore, we give an exact formula for the number of fixed-length words that do not contain the codewords in a variable-length non-overlapping code as subwords. This thereby solves an open problem by Bilotta and induces a recursive upper bound on the maximum size of variable-length non-overlapping codes.
\end{abstract}


\section{Introduction}

Motivated by applications for synchronization, non-overlapping codes were first defined by Levenshtein in 1964 under the name \emph{strongly regular codes}, a.k.a. \emph{codes without overlaps}~\cite{levenshtein1964decoding},~\cite{levenshtein1970maximum},~\cite{levenshtein2004combinatorial}.
A code $S\subseteq \bigcup_{n \ge 2} \mathbb{Z}_q^n$ is called \emph{non-overlapping} if $S$ satisfies the following two conditions:
\begin{enumerate}
    \item[(1)] No non-empty prefix of each codeword is a suffix of any one, including itself;
    \item[(2)] For all distinct $\mathbf{u,v} \in S$, $\mathbf{u}$ does not contain $\mathbf{v}$ as a subword.
\end{enumerate}

We say that $S$ is a fixed-length non-overlapping code if $S \subseteq \mathbb{Z}_q^n$, otherwise it is called a variable-length non-overlapping code.
In this paper, we consider both fixed-length and variable-length cases.
Fixed-length non-overlapping codes have been intensively studied in the literature. 
Let $C(n,q)$ be the maximum size of a $q$-ary non-overlapping codes of length $n$.
The main research problems are to construct non-overlapping codes as large as possible in size and to bound $C(n,q)$. 
The first construction was proposed by Levenshtein in 1964~\cite{levenshtein1964decoding},~\cite{levenshtein1970maximum} (Construction~I, see also~\cite{gilbert1960synchronization},~\cite{chee2013cross}).
Following the work by de Lind van Wijngaarden and Willink~\cite{van2000frame} in 2000,
Bajic and Stojanovic~\cite{bajic2004distributed} independently rediscovered binary fixed-length non-overlapping codes (under the name \emph{cross-bifix-free codes}) in 2004.
In 2012, Bilotta, Pergola, and Pinzani~\cite{bilotta2012new} provided a binary construction based on Dyck paths, by which the code size is smaller than Levenshtein's.
However, it reveals an interesting connection between non-overlapping codes and other combinatorial objects.
In 2013, Chee, Kiah, Purkayastha, and Wang~\cite{chee2013cross} rediscovered Levenshtein's construction (Construction~I), and verified that it is optimal for $q=2$ and $n \le 16$, expect when $n = 9$ by computer search. 
In 2015, Blackburn~\cite{blackburn2015non} generalized Levenshtein's construction and thereby provided a class of largest fixed-length non-overlapping codes (see Construction~I.A) whenever $n$ divides $q$.
In 2016, Barucci, Bilotta, Pergola, Pinzani, and Succi~\cite{barcucci2016cross} proposed a construction for $q$-ary ($q \ge 3$) fixed-length non-overlapping code based on colored Motzkin paths.

The best known lower bound of $C(n,q)$ for fixed $q$, proposed by Levenshtein~\cite{levenshtein1964decoding} (see also~\cite{chee2013cross},~\cite{gilbert1960synchronization},~\cite{levy2018mutually}), states that,
\begin{equation*}
    C(n,q)  \gtrsim \frac{q-1}{qe}\frac{q^n}{n},
\end{equation*} 
as $n \to \infty$ over the subsequence $n = \frac{q^i - 1}{q-1}$ for $i = 1,2,\dots$, where $e$ is the base of natural logarithm. 
This lower bound is derived from the cardinality of non-overlapping codes by Construction~I.
The best known upper bound, also by Levenshtein~\cite{levenshtein1970maximum}, states that 
\begin{equation}\label{equ: lev bound}
    C(n,q) \le \left(\frac{n-1}{n}\right)^{n-1} \frac{q^n}{n}.
\end{equation}
Blackburn~\cite{blackburn2015non} showed that this bound is tight if $n$ divides $q$. 
A weaker bound 
\begin{equation}\label{equ: com bound}
    C(n,q) \le \frac{q^n}{2n-1}
\end{equation}
was found by Chee \textit{et al.}~\cite{chee2013cross} independently, and Blackburn~\cite{blackburn2015non} further proved the equality cannot hold via a simple proof.
Moreover, Blackburn~\cite{blackburn2015non} showed that for fixed $n \ge 2$, 
\begin{equation*}
    \liminf_{q \to \infty}\frac{C(n,q)}{q^n} = \frac{1}{n} \left(\frac{n-1}{n}\right)^{n-1},
\end{equation*}
and there exist absolute constants $c_1$, $c_2$ such that $c_1 (q^n/n) \le C(n,q) \le c_2(q^n/n)$ for all $n,q$ larger than $1$.

Recently, non-overlapping codes have also found applications in DNA storage systems~\cite{levy2018mutually},~\cite{yazdi2015rewritable}, pattern matching~\cite{crochemore2007algorithms}, and automata theory~\cite{berstel2010codes}.
For other related work on fixed-length non-overlapping codes, see also~\cite{levy2018mutually},~\cite{bajic2014simple},~\cite{bernini2014prefix},~\cite{bernini2017gray}.

Variable-length non-overlapping codes, defined several decades ago, were not further studied until recently.
In 2017, Bilotta~\cite{bilotta2017variable} constructed a class of binary variable-length non-overlapping codes (see Construction~II) and gave a recurrence relation on their cardinalities.
It was proved that the cardinality is about $[2(1 - \epsilon)]^n$ as $n \to \infty$, where $\epsilon$ is a small positive value.
In addition, a recursive bound on the maximum size of binary variable-length non-overlapping codes was proposed. 
More precisely, let $\mathcal{J}_2 = \bigcup_{n \ge 2} J_2(n)$ be a binary variable-length code, where $J_2(n) \subseteq \mathbb{Z}_2^n$.
Denote by $b_{\mathcal{J}_2(n)}(i)$ the number of $i$-length binary words that do not contain codewords in $\mathcal{J}_2$ as subwords.
It was shown that
\begin{equation}\label{equ: bilotta's upper bound}
    |J_2(n)| < \frac{2^n}{n+1} - \sum_{i=1}^{h-2}2^i |J_2(n-i)| - \frac{1}{2}\sum_{i=h}^{n+1-h} b_{\mathcal{J}_2(n)}(i) \cdot |J_2(n+1-i)|.
\end{equation}
However, the exact expression for $b_{\mathcal{J}_2(n)}(i)$ was left open by Bilotta~\cite{bilotta2017variable}. 

Variable-length non-overlapping codes could also be used in DNA storage systems as the address sequences~\cite{yazdi2015rewritable}. 
More importantly, for two codewords with length $m$ and $m+d$ respectively, their Hamming distance is admitted at least $d$. 
This is the major advantage of using variable-length non-overlapping codes for DNA storage systems.

The main contribution of this paper is summarized as follows. 
Firstly, a generic method to generalize binary non-overlapping codes to $q$-ary is provided;
secondly, a generating function approach is utilized to analyze the generalized constructions of both Levenshtein's and Bilotta's. More precisely, for Construction~I' (the generalized Levenshtein's construction), we find a $2$-term recurrence relation of the cardinalities, and it reveals a new combinatorial interpretation for the structure of this classic construction. 
In addition, the asymptotic behavior of the cardinality of codes by Construction~I' for a special case is also given.
Numerical results suggest that Construction~I' may outperform other constructions of fixed-length non-overlapping codes, especially when $q$ is large.
Meanwhile, for Construction~II' (the generalized Bilotta's construction), we give the generating function and a $3$-term recurrence relation of the cardinalities which cannot be trivially obtained by Bolotta's method.
Finally, an exact formula of  $b_{\mathcal{J}_2(n)}(i)$ is given and Eq.~\eqref{equ: bilotta's upper bound} is further generalized to the $q$-ary case.

The rest of this paper is organized as follows. 
In Section \ref{sec: notations and definitions}, we provide some necessary notations and definitions.
In Section \ref{sec: from binary to qary}, we introduce the generic method to generalize binary codes to $q$-ary, and investigate how it works on the aforementioned two constructions.
In Section \ref{sec: upper bounds on variable-length non-overlapping codes}, we give an exact formula of $b_{\mathcal{J}_2(n)}(i)$ and generalize the recursive upper bound Eq.~\eqref{equ: bilotta's upper bound} for $q$-ary non-overlapping codes.
Finally, we conclude this paper with some open problems in Section \ref{sec: conclusions}.

\section{Notations and definitions} \label{sec: notations and definitions}

In this paper, both $q$-ary codewords and words are vectors over $\mathbb{Z}_q = \{0, 1,\dots, q-1\}$. We use $[n]$ to denote the set $\{1, \ldots, n\}$ for an integer $n \ge 1$. It is convenient to write vectors as strings. For example, $02101$ represents the vector $(0,2,1,0,1)$. A $q$-ary code is a set of codewords over $\mathbb{Z}_q$ and is called \emph{variable-length} if its codewords have different lengths. The size of a $q$-ary code $S$ is the number of codewords in $S$, and is denoted by $|S|$.
The empty set is denoted by $\emptyset$.

\begin{definition}\label{def: prefix and suffix}
    Let $n,q$ be integers larger than $1$. For $\mathbf{u} \in \mathbb{Z}_q^n$, the set of prefixes of $\mathbf{u}$ is  $\pre (\mathbf{u}) = \{ (u_1, u_2, \dots u_i) \mid 1 \le i \le n-1\}$, and the set of suffixes of $\mathbf{u}$ is $\suf(\mathbf{u}) = \{ (u_i, u_{i+1}, \dots, u_n) \mid 2 \le i \le n \}$. 
\end{definition}

For example, $\pre(0011) = \{0,00,001\}$ and $\suf(0011) = \{1, 11,011\}$.

\begin{definition}[Non-overlapping codes]\label{def: non-overlapping codes}
    A $q$-ary non-overlapping code is a finite or countable code $S \subseteq \bigcup_{n \ge 2} \mathbb{Z}_q^n$ which satisfies the following two conditions:
    \begin{enumerate}
        \item[(1)] For all $\mathbf{u},\mathbf{v} \in S$, $\pre(\mathbf{u}) \cap \suf(\mathbf{v}) = \emptyset$;
        \item[(2)] For all distinct $\mathbf{u,v} \in S$, $\mathbf{u}$ does not contain $\mathbf{v}$ as a subword.
    \end{enumerate} 
\end{definition}

For example, $\{00101,00111\}$, $\{11101000, 111011000\}$ are both binary non-overlapping codes, while $\{0111, 0011\}$ is not since $011 \in \pre(0111) \cap \suf(0011)$.
In addition, $\{10, 1{10}0\}$ is not a non-overlapping code since $10$ is a subword of $1100$.

\section{A generic method to generalize binary codes to $q$-ary } \label{sec: from binary to qary}
We now define a generic way to generalize binary codes to $q$-ary codes.

\begin{definition} \label{def: binary to q-ary codes}
    Let $I,J$ form a bipartition of $\mathbb{Z}_q$ with $q \ge 2$ and $\mathbf{\omega} = (\omega_1, \dots, \omega_n) \in \mathbb{Z}_2^n$. 
    We define a map $\phi_{I,J}(\cdot)$ from a binary codeword to a set of $q$-ary codewords as follows.
    \[
            \phi_{I,J}(\mathbf{\omega}) = A_1 \times A_2 \times \cdots \times A_{n},
    \]
    where $\times$ is the Cartesian product and for $1 \le i \le n$,
    \[
        A_i =
        \begin{cases}
            I & \text{if } {\omega}_i = 0, \\ 
            J & \text{if } {\omega}_i = 1.    
        \end{cases}
    \]
    Define the map $\Phi_{I,J}(\cdot)$ from a binary code $S \subseteq \bigcup_{n \ge 2} \mathbb{Z}_2^n$ to a $q$-ary code as
    \[
    \Phi_{I,J}(S) = \bigcup_{\mathbf{\omega} \in S} \phi_{I,J}(\mathbf{\omega}).
    \]
\end{definition}

\begin{example}\label{ex: binary to q-ary words}
    Let $I = \{0,2\}, J = \{1,3\}$. Then 
    \[
    \phi_{I,J}(001) = \{0,2\} \times \{0,2\} \times \{1,3\}
    \]
    is a code over $\mathbb{Z}_4$ with length $3$ and size $8$.
\end{example}

The non-overlapping property is preserved by the map we defined above.
We give the following lemma.
\begin{lemma}\label{corollary: binary non-overlapping to q-ary non-overlapping}
    Let notations be the same as above.
    If $S$ is a binary non-overlapping code,
    then $\Phi_{I,J}(S)$ is a $q$-ary non-overlapping code.
\end{lemma}

\begin{proof}
    Suppose to the contrary that $\Phi_{I,J}(S)$ is overlapping. 
    Then there are two possible cases:
    \begin{enumerate}
        \item[Case 1):] There exist two $q$-ary codewords $\mathbf{a}' = (a_1, \dots, a_{n_1})$ and $\mathbf{b'} = (b_1, \dots, b_{n_2}) $ from $ \Phi_{I,J}(S)$ such that $\pre(\mathbf{a}') \cap \suf(\mathbf{b}') \ne \emptyset$. 
        W.l.o.g., assume that $(x_1,\dots,x_k) \in \pre(\mathbf{a}') \cap \suf(\mathbf{b}')$ for a certain integer $k$. 
        Thus we have $({a}'_1, \dots, {a}'_k) = ({b}'_{n_2 - k + 1}, \dots, {b}'_{n_2})$. 
        By Definition \ref{def: binary to q-ary codes}, there exist two $q$-ary codewords $\mathbf{a} = (a_1, \dots, a_{n_1}) ,\mathbf{b} = (b_1,\dots,b_{n_2}) $from $ S$ such that $\mathbf{a}' \in \phi_{I,J}(\mathbf{a})$ and $\mathbf{b}' \in \phi_{I,J}(\mathbf{b})$. It then follows that $({a}_1, \dots, {a}_k) = ({b}_{n_2 - k + 1}, \dots, {b}_{n_2})$. 
        This is impossible since $S$ is non-overlapping.
    
        \item[Case 2):] There exist two $q$-ary codewords $\mathbf{a}',\mathbf{b'} $ from $ \Phi_{I,J}(S)$ such that $\mathbf{a}'$ contains $\mathbf{b'}$ as a subword.
        By Definition \ref{def: binary to q-ary codes}, there exist two $q$-ary codewords $\mathbf{a},\mathbf{b} $ from $S$ such that $\mathbf{a}' \in \phi_{I,J}(\mathbf{a})$, $\mathbf{b}' \in \phi_{I,J}(\mathbf{b})$, and $\mathbf{a}$ contains $\mathbf{b}$ as a subword.
        This is again impossible since $S$ is non-overlapping.
    \end{enumerate}
    The proof is then completed.
\end{proof}

In the following, we investigate how the generalization in Definition~\ref{def: binary to q-ary codes} works on two known constructions. 
Lemma~\ref{lemma: root of denominator} will be useful to analyze the cardinalities of the codes by the two generalized constructions.

\begin{lemma}\cite[Lemma~3.6, Theorem~3.9]{wolfram1996solving}\label{lemma: root of denominator}
    Let $k$ be an integer larger than $1$.
    The equation $y^k - \sum_{i=0}^{k-1}y^i =0$ has one positive real root $y_0$ in the open interval $(1,2)$. If $k$ is even, it also has a negative real root $y_1$ belonging to $(-1,0)$. Moreover, $y_0 = 2(1-\epsilon_k) > 2(1-2^{-k})$ and approaches to $2$ as $k \to \infty$, where 
    \begin{equation}\label{equ: epsilon_k}
        \epsilon_k = \sum_{i \ge 1} \binom{(k+1)i - 2}{i - 1} \frac{1}{i 2^{(k+1)i}}.    
    \end{equation}
    The root $y_1$ and each complex root have modulus lying in the open interval $(3^{-k},1)$. 
\end{lemma}

\subsection{The first construction} \label{sec: non-overlapping codes}
We first review the classic construction given by Levenshtein~\cite{levenshtein1964decoding},~\cite{levenshtein1970maximum}, which was also considered by Gilbert~\cite{gilbert1960synchronization}, and rediscovered by Chee \textit{et al.}~\cite{chee2013cross}. 

\begin{construction_classic}[\cite{levenshtein1964decoding},~\cite{levenshtein1970maximum},~\cite{gilbert1960synchronization},~\cite{chee2013cross}]\label{cons: classic}
Let $n,q$ be integers larger than $1$ and $1 \le k \le n-1$. Denote by $S_q^{(k)}(n)$ the set of all codewords $\mathbf{s} = (s_1, s_2, \dots, s_n) \in \mathbb{Z}_q^n$ such that
\begin{itemize}
    \item $s_1 = s_2 = \dots = s_k = 0$, and $s_{k+1} \ne 0$, $s_n \ne 0$;
    \item $(s_{k+1}, \dots, s_{n})$ does not contain $k$ consecutive $0$'s.
\end{itemize}
Then $S_q^{(k)}(n)$ is a non-overlapping code.
\end{construction_classic}

\begin{example}\label{ex: ConstructionI}
    Take $k=2,q=2$ and $n=6$. Construction~I gives the following non-overlapping code.
    \[
    S_2^{(2)}(6) = 
    \{
    001011, 001101,001111     
    \}.    
    \]
\end{example}

Chee \textit{et al.}~\cite{chee2013cross} verified that this construction gives largest possible non-overlapping codes for $q=2$ and $n \le 16$, expect for $n = 9$.
Let $C$ be a subset of $\mathbb{Z}_q^k$, we say that $(x_1,x_2, \dots, x_r) \in \mathbb{Z}_q^r$ is \emph{$C$-free} if $r < k$, or if $r \ge k$, $(x_i,x_{i+1}, \dots,x_{i+k-1}) \notin C$ holds for all $i = 1,2,\dots, r-k+1$.
To obtain non-overlapping codes as large as possible in size, Construction~I was generalized by Blackburn~\cite{blackburn2015non} as follows.

\begin{construction_Blackburn}[\cite{blackburn2015non}]\label{cons: blackburn}
    Let $n,q$ be integers larger than $1$ and $1 \le k \le n-1$. Let $I,J$ form a bipartition of $\mathbb{Z}_q$, and $C \subseteq I^k$. Denote by $S_{I,J}^{C}(n)$ the set of all codewords $\mathbf{s} = (s_1, s_2, \dots,s_n) \in \mathbb{Z}_q^n$ such that
    \begin{itemize}
        \item $(s_1,s_2,\dots,s_k) \in C$, and $s_{k+1} \in J$, $s_n \in J$;
        \item $(s_{k+1}, \dots, s_{n})$ is $C$-free.
    \end{itemize}
    Then $S_{I,J}^{C}(n)$ is a non-overlapping code.
\end{construction_Blackburn}

Blackburn~\cite{blackburn2015non} showed that $S_{I,J}^{I}(n) = I \times J^{n-1}$ achieves Levenshtein's bound (see Eq.~\eqref{equ: lev bound}) if $n$ divides $q$ and $|I| = q/n$.
However, the cardinality of $S_{I,J}^{C}(n)$ for general $C$ is hard to decide and left open.

In the following, we give a generalization of Construction~I by the method in Definition~\ref{def: binary to q-ary codes}.
We remark that this is a special case of Construction~I.A.
However, it is worth noting that this leads to more transparent analysis on the cardinality using a generating function approach.
Recalling that Chee \textit{et al.}~\cite{chee2013cross} showed $S_q^{(k)}(n)$ is non-overlapping for $q \ge 2$, then by Lemma~\ref{corollary: binary non-overlapping to q-ary non-overlapping}, $\Phi_{I,J}(S_2^{(k)}(n))$ is also non-overlapping.
It can be easily verified that $\Phi_{I,J}(S_2^{(k)}(n))$ is equivalent to the following construction. 
\begin{construction_Phi_classic}\label{cons: Phi classic}
    Let $n,q$ be integers larger than $1$ and $1 \le k \le n-1$. Let $I,J$ form a bipartition of $\mathbb{Z}_q$. Denote by $S_{I,J}^{(k)}(n)$ the set of all codeword $\mathbf{s} = (s_1, s_2, \dots, s_n) \in \mathbb{Z}_q^n$ such that
    \begin{itemize}
        \item $(s_1,s_2,\dots, s_k)\in I^k$, and $s_{k+1} \in J$, $s_n \in J$;
        \item $(s_{k+1}, \dots, s_{n})$ is $I^k$-free.
    \end{itemize}
    Then $S_{I,J}^{(k)}(n)$ is a non-overlapping code.
\end{construction_Phi_classic}

\begin{remark}\label{remark: phi classic}
    Construction~I can be viewed as a special case of Construction~I'.
    More precisely, taking $I = \{0\}$, $J = \{1,2,\dots,q-1\} = [q-1]$ in Construction~I', we have $S_{\{0\},[q-1]}^{(k)} (n)= S_q^{(k)}(n) = \Phi_{\{0\},[q-1]}(S_2^{(k)}(n))$; 
    Construction~I' can be viewed as a special case of Construction~I.A. 
    More precisely, taking $C = I^k$ in Construction~I.A, we have $S_{I,J}^{I^k}(n) = S_{I,J}^{(k)}(n)$. 
    Hence, $S_{I,J}^{(1)}(n)$ also achieves Levenshtein's bound (see Eq.~\eqref{equ: lev bound}) when $n$ divides $q$ and $|I| = q/n$.
\end{remark}

A fixed-length non-overlapping code $S \subseteq \mathbb{Z}_q^n$ is called {\em non-expandable} if $S \cup \{\mathbf{x}\}$ is overlapping for each $\mathbf{x} \in \mathbb{Z}_q^n \setminus S$.
A non-expandable non-overlapping code must be maximal but may not be of largest size.
In~\cite[Theorem~3.1]{chee2013cross}, it is claimed that $ S_q^{(k)}(n)$ is non-expandable for all $2 \le k \le n-2$. 
However, a key condition was ignored in the proof and the statement was thereby incorrect.
To see this, by Construction~I, $S_2^{(3)}(6) = \{000101,000111\}$, but $S_2^{(3)}(6) \cup \{001101\}$ is still non-overlapping. 
Similarly, $S_2^{(4)}(7) = \{0000101,0000111\}$ but $S_2^{(4)}(7) \cup \{0001001\}$ is still non-overlapping.
In the following, we give a generalized result on the non-expandability of the codes by Construction~I'.

\begin{theorem}\label{thm: non-expandability of ConstructionI'}
    The $q$-ary non-overlapping code $S_{I,J}^{(k)}(n)$ by Construction~I' is non-expandable if $k=n-1$ or $1 \le k < n/2$. 
\end{theorem}

\begin{proof}
    By Construction~I', the $q$-ary word $\mathbf{x} = (x_1, \dots, x_n)$ overlaps with every codeword from $S_{I,J}^{(k)}(n)$ if $x_1 \in J$ or $x_n \in I$.
    Therefore, we only need to consider $\mathbf{x} \in (I \times \mathbb{Z}_q^{n-2} \times J) \setminus S_{I,J}^{(k)}(n)$.
    Note that $S_{I,J}^{(n-1)}(n) = I^{n-1} \times J$ is clearly non-expandable.
    For $1 \le k < n/2$, there are two cases:
    \begin{enumerate}
        \item[Case 1):] The word $\mathbf{x}$ is $I^k$-free. 
        In this case, it belongs to $I^{l} \times J \times \mathbb{Z}_q^{n-l-1}$ for some $0 \le l \le k-1.$ 
        The $(l+1)$-length prefix of $\mathbf{x}$ is also a suffix for some $\mathbf{y}$ in $I^k \times J^{n - k -l -1} \times I^l \times J \subseteq S_{I,J}^{(k)}(n)$ since $n > 2k \ge k + l + 1$.
        Hence, $S_{I,J}^{(k)}(n) \cup \{ \mathbf{x} \}$ is overlapping.

        \item[Case 2):] The word $\mathbf{x}$ contains $\mathbf{\omega} \in I^k$ as a subword.
        We consider the shortest suffix of $\mathbf{x}$ which is not $I^k$-free.
        In this case, $\mathbf{x}$ belongs to $\mathbb{Z}_q^{n-m-k-1} \times I^k\times J \times T^{(k)}(m)$, where $T^{(k)}(m)$ is the set of $q$-ary $I^k$-free words of length $m < n - k - 1$.
        It then follows that the $(k+1+m)$-length suffix of $\mathbf{x}$ is also a prefix for some $\mathbf{y}$ in $I^k \times J \times T^{(k)}(m) \times J^{n - m - k -1} \subseteq S_{I,J}^{(k)}(n)$.
        Hence, $S_{I,J}^{(k)}(n) \cup \{ \mathbf{x} \}$ is overlapping.
    \end{enumerate}
    Therefore, no additional word $\mathbf{x}$ can be appended to the set $S_{I,J}^{(k)}(n)$ such that $S_{I,J}^{(k)}(n) \cup \{ \mathbf{x} \}$ is still non-overlapping.
    The proof is then completed.
\end{proof}

\begin{remark}
    In the proof of~\cite[Theorem~3.1]{chee2013cross}, the condition $n- k -l -1 > 0$ in Case 1) of the proof of Theorem~\ref{thm: non-expandability of ConstructionI'} was ignored.
\end{remark}

We now clarify the result of~\cite[Theorem~3.1]{chee2013cross} in the following Corollary.

\begin{corollary}\label{corollary: chee non-expandable}
    The code $S_q^{(k)}(n)$ by Construction~I is non-expandable if $k = n-1$ or $1 \le k < n/2$.
\end{corollary}

In the rest of this subsection, we analyze the cardinality of codes by Construction~I' using a generating function approach.
Let $S_{I,J}^{(k)}(n+k)$ be a code defined by Construction~I', and $U_{I,J}^{(k)}(n)$ be the set of $I^k$-free $q$-ary codewords with length $n$ that start and end both with an element in $J$.
Note that $S_{I,J}^{(k)}(n+k)= I^k \times U_{I,J}^{(k)}(n)$, $U_{I,J}^{(1)}(n) = J^n$, $U_{I,J}^{(k)}(1) = J$ and $U_{I,J}^{(k)}(2) = J \times J$.
Therefore, the cardinality of $S_{I,J}^{(k)}(n+k)$ is determined by that of $U_{I,J}^{(k)}(n)$.
Define
\begin{equation}\label{equ: u_{I,J}}
    u_{I,J}^{k}(n) = 
    \begin{cases}
        0 & \text{if }n \le 0, \\
        |U_{I,J}^{k}(n)| & \text{if } n \ge 1,
    \end{cases}
\end{equation}
then we have the following lemma.

\begin{lemma} \label{lemma: ofg of J*J}
    Let $I,J$ form a bipartition of $\mathbb{Z}_q$ and $u_{I,J}^{(k)}(n)$ be defined in~\eqref{equ: u_{I,J}}. For any fixed $k \ge 1$, we have
    \begin{equation} \label{equ: ogf of J*J}
        \sum_{n=0}^{\infty} u_{I,J}^{(k)}(n) x^n = \frac{|J|x (1 - |I|x)}{1 - qx + |I|^k|J|x^{k+1}},
    \end{equation}

    Therefore, we have
    \begin{equation}\label{equ: recurrence relation on J*J}
        u_{I,J}^{(k)}(n) = 
    \begin{cases}
        0  & \text{if } n \le 0, \\
        |J|^n  & \text{if } n = 1,2, \\
        q \cdot u_{I,J}^{(k)}(n-1) - |I|^k|J|u_{I,J}^{(k)}(n-k-1) & \text{if } n \ge 3.
    \end{cases}
    \end{equation}
\end{lemma}

\begin{proof}
    We first show that Eq.s~\eqref{equ: ogf of J*J} and~\eqref{equ: recurrence relation on J*J} hold for $k\ge 2$, then we consider the remaining case $k=1$.
    Let $k \ge 2$.
    By definition, each codeword from $\bigcup_{n \ge 1}U_{I,J}^{(k)}(n) $ can be uniquely constructed by appending letters from the alphabet $\mathbb{Z}_q$ to the empty word as follows. 
    \begin{enumerate}
        \item [(1)] Repeat the following two steps for $i \ge 0$ times;
        \begin{enumerate}
            \item[(a)] append arbitrarily many consecutive letters from $J$;
            \item[(b)] append less than $k$ consecutive letters from $I$;
        \end{enumerate}
        \item[(2)] Append arbitrarily many consecutive letters from $J$.
    \end{enumerate}

    Therefore, we have 
    
    \[  
        \begin{aligned}
            \sum_{n=0}^{\infty} u_{I,J}^{(k)}(n) x^n  & =  \left \{  \sum_{i=0}^{\infty} \left[ (|J|x + |J|^2x^2 + \dots) (|I|x + |I|x^2 + \dots + |I|^{k-1}x^{k-1}) \right] ^i\right \}  \cdot  (|J|x + |J|^2x^2 + \dots) \\
            & = \frac{|J|x}{1 - |J|x} \sum_{i=0}^{\infty} \left[\left( \frac{|J|x}{1 - |J|x} \right) \left( \frac{|I|x - |I|^k x^k}{1 - |I|x} \right) \right] ^i \\
            & = \frac{\frac{|J|x}{1 - |J|x}}{1 - \left( \frac{|J|x}{1 - |J|x} \right) \left( \frac{|I|x - |I|^k x^k}{1 - |I|x} \right)} \\
            & = \frac{|J|x}{1 - |J|x} \cdot \frac{(1- |J|x) (1 - |I|x)}{(1- |J|x) (1 - |I|x) - |J|x (|I|x - |I|^k x^k)} \\
            & = \frac{|J|x (1 - |I|x)}{1 - (|I| + |J|)x + |I|^k|J|x^{k+1}}.
        \end{aligned}
    \]
    Then Eq.~\eqref{equ: ogf of J*J} follows from the relation that $|I| + |J| = q$. 

    Next we show Eq.~\eqref{equ: recurrence relation on J*J} holds when $k \ge 2$.
    For $n \le 2$, Eq.~\eqref{equ: recurrence relation on J*J} follows from Eq.~\eqref{equ: u_{I,J}}.
    For $n \ge 3$, multiplying $1 - qx + |I|^k|J|x^{k+1}$ to both sides of Eq.~\eqref{equ: ogf of J*J} and comparing the coefficient of $x^n$, we get Eq.~\eqref{equ: recurrence relation on J*J}.

    Finally, let $k = 1$, then we have
    \begin{equation} \label{equ: recurrence relation on J^*}
        u_{I,J}^{(1)}(n) = 
        \begin{cases}
            0  & \text{if } n \le 0, \\
            |J|^n & \text{if } n \ge 1.
        \end{cases}
    \end{equation}
    
    It follows that

    \[
        \sum_{i=0}^{\infty} u_{I,J}^{(1)}(i)x^i = |J|x + |J|^2x^2 + \dots = \frac{|J|x}{1 - |J|x}.   
    \]
    Clearly, Eq.~\eqref{equ: recurrence relation on J^*} coincides with Eq.~\eqref{equ: recurrence relation on J*J} if $k$ = 1 and $n \le 2$. For $n \ge 3$, note that Eq.~\eqref{equ: recurrence relation on J*J} yields $u_{I,J}^{(1)}(n) = |I|\left(u_{I,J}^{(1)}(n-1) - |J|u_{I,J}^{(1)}(n-2)\right) + |J|u_{I,J}^{(1)}(n-1) = |J|u_{I,J}^{(1)}(n-1)$. Hence, Eq.~\eqref{equ: recurrence relation on J^*} coincides with Eq.~\eqref{equ: recurrence relation on J*J} for all $n \ge 0$ if $k$ = 1.

    Setting $k=1$ in Eq.~\eqref{equ: ogf of J*J}, we have 
    \[
        \frac{|J|x(1 - |I|x)}{1 - qx + |I||J|x^2} = \frac{|J|x}{1 - |J|x}.    
    \]

    The proof is then completed.
\end{proof}

\begin{theorem}\label{thm: Phi chee}
    Let $S_{I,J}^{(k)}(n)$ be the code given by Construction~I' with $k \ge 1$. 
    Then we have

    \begin{equation} \label{equ: ogf of I^kJ*J}
        \sum_{n=0}^{\infty} |S_{I,J}^{(k)}(n)|x^n = \sum_{n=0}^{\infty} |I|^k u_{I,J}^{(k)}(n-k)x^n
              = \frac{|I|^{k}|J|x^{k+1}(1 - |I|x)}{1 - qx + |I|^k |J| x^{k+1}},
    \end{equation}
    
    and
    
    \begin{equation}\label{equ: recurrence relation of I^kJ*J}
        |S_{I,J}^{(k)}(n)| = 
        \begin{cases}
            0 & \text{if } n \le k, \\
            |I|^k|J|^{n-k}  & \text{if } n = k+1, k+2, \\
            q \cdot |S_{I,J}^{(k)}(n-1)| - |I|^k|J|\cdot |S_{I,J}^{(k)}(n-k- 1)| & \text{if } n \ge k + 3.
        \end{cases}
    \end{equation}
\end{theorem}
\begin{proof}
    From the fact that $S_{I,J}^{(k)}(n) = I^k \times U_{I,J}^{(k)}(n-k)$ and Lemma~\ref{lemma: ofg of J*J}, the conclusion follows.
\end{proof}

We remark that Eq.~\eqref{equ: recurrence relation of I^kJ*J} has a direct combinatorial interpretation.
Let $n \ge k+3$. 
Define $P \subseteq \mathbb{Z}_q^n$ such that for any $\mathbf{p} = (p_1,p_2,\dots,p_n) \in P$, $(p_1,p_2,\dots, p_{n-2},p_n) \in S_{I,J}^{(k)}(n-1)$ and $p_{n-1} \in \mathbb{Z}_q$. 
In other words, $P$ is constructed by inserting an arbitrary $q$-ary letter in the second last position for each codeword in $S_{I,J}^{(k)}(n-1)$.
Let $Q = S_{I,J}^{(k)}(n)$. 
Let $T = S_{I,J}^{(k)}(n-k-1)\times I^k \times J$.
Note that $Q \subseteq P, P \setminus Q = T$, $|P| = q \cdot |S_{I,J}^{(k)}(n-1)|, |Q| = |S_{I,J}^{(k)}(n)|$, and $|T| = |I|^k|J|\cdot |S_{I,J}^{(k)}(n-k- 1)|$. Eq.~\eqref{equ: recurrence relation of I^kJ*J} then follows from $|P| - |Q| = |T|$.

Chee \textit{et al}.~\cite{chee2013cross} determined $|S_q^{(k)}(n)| $ by giving a recurrence relation of the set $S_q^{(k)}(n)$ in $k$ terms of $S_q^{(k)}(n-l)$ for $1\le l \le k$. 
Hence we give a simpler recurrence of $S_q^{(k)}(n)$ involving only two terms: $S_q^{(k)}(n-1)$ and $S_q^{(k)}(n-k-1)$. 
The generating function approach is new and thereby provides us a more transparent understanding of this classic construction.
It can be verified that Theorem~\ref{thm: Phi chee} coincides with Chee \textit{et al}'s result~\cite[Corollary~3.1]{chee2013cross} given that $I = \{0\}$ and $J = [q-1]$.

In the following, we estimate the asymptotic behavior of the cardinality of $S_{I,J}^{(k)}(n)$ when $|I| = |J|$.
Clearly, $|S_{I,J}^{(1)}(n)| = (q/2)^n$. 
For $k > 1$, we give the following theorem.
\begin{theorem}\label{thm: upper bound of ConstructionI}
    Let $S_{I,J}^{(k)}(n)$ be the code given by Construction~I' with $|I| = |J|$ and $k > 1$, we have
    \begin{equation*}\label{equ: Phi classic size}
        \limsup_{n \to \infty} |S_{I,J}^{(k)}(n)| ^{1/n} =  q(1 - \epsilon_k),
    \end{equation*}
    where $\epsilon_k$ is given in Eq.~\eqref{equ: epsilon_k}.
\end{theorem}

\begin{proof}
    It is well known in complex analysis (see~\cite[Theorem 2.5]{stein2010complex}, for example) that for a power series $f = \sum_{n \ge 0} a_n x^n$, the limit superior of the sequence $\{|a_n|\}_n^{1/n}$ is $R^{-1}$ as $n \to \infty$, where $R$ is the radius of convergence of $f$. 

    Define $F(x) = \sum_{n=0}^{\infty} |S_{I,J}^{(k)}(n)|x^n$. 
    By Eq.~\eqref{equ: ogf of I^kJ*J}, we have
    
    \[
        F(x) = \frac{(rx)^{k+1}(1-rx)}{1 - 2rx + (rx)^{k+1}} = \frac{(rx)^{k+1}}{1 - \sum_{i=1}^{k}(rx)^i},
    \]
    where $r = q/2$.

    The coefficient of $F(x)$ are all non-negative and by Pringsheim's Theorem (see~\cite[Theorem.IV.6]{flajolet2009analytic}, for example),
    the radius of convergence of $F(x)$ (denoted by $R$) is the smallest real root of the denominator.
    Let $f(x) = 1 - \sum_{i=1}^k(rx)^i$, then $f(R)=0$ and $0 < R < 1/r$ since $f(0) > 0$ and $f(1/r) < 0$. 
    Let $y = rx$, and $f(x) = z(y) = 1 - \sum_{i=1}^{k}y^i$.
    The reciprocal polynomial of $z(y)$ is 
    
    \[
        z^*(y) = y^k - y^{k-1} - \dots - y - 1.
    \]
    
    By Lemma~\ref{lemma: root of denominator}, $z^*(2(1 - \epsilon_k)) = 0$, and $\left[ 2(1-\epsilon_k)\right]^{-1}$ is the unique positive real root of $z(y)$. Therefore $R = \left[ 2r(1-\epsilon_k)\right]^{-1}$ and we have
    
    \[
        \limsup_{n \to \infty} |S_{I,J}^{(k)}(n)|^{1/n} = R^{-1} = q(1 - \epsilon_k),
    \]    
    where $\epsilon_k$ is given in Eq.~\eqref{equ: epsilon_k}.
    The proof is then completed.
\end{proof}

By Theorem~\ref{thm: upper bound of ConstructionI}, we know that, for $|I| = |J|$, $k>1$ and all fixed $\epsilon >0$, there exists an integer $N$ such that $|S_{I,J}^{(k)}(n)| < \left[ q(1 - \epsilon_k) + \epsilon \right]^n$ holds for all $n > N$, and $|S_{I,J}^{(k)}(n)| > \left[ q(1 - \epsilon_k) - \epsilon \right]^n$ holds for infinitely many values of $n$.

Recall in Remark~\ref{remark: phi classic} that Construction~I' generates largest non-overlapping code when $n$ divides $q$. 
We further note that the code by Construction~I' is close to Levenshtein's bound (see Eq.~\eqref{equ: lev bound}) even if $n$ does not divide $q$.
Let $S_{I,J}^{(1)}(n) = I \times J^{n-1}$ with $|I| = \lfloor q/n \rfloor$ be the code by Construction~I'.
Then by $|J| = q - |I| \ge q (\frac{n-1}{n})$, we have 
\[
|S_{I,J}^{I}(n)| = \frac{1}{n} \left( \frac{n-1}{n} \right)^{n-1}q^{n} - O(q^{n-1}).
\]

\subsection{The second construction} \label{sec: varianle-length non-overlapping codes}
Before we give the generalization on variable-length non-overlapping codes, we first review the binary construction by Bilotta~\cite{bilotta2017variable}.

\begin{construction_Bilotta}[Bilotta~\cite{bilotta2017variable}]\label{cons: Bilotta}
    Let $n, k$ be integer such that $3 \le k \le \lfloor n/2 \rfloor -2$. 
    Denote by $V_2^{(k)}(n)$ the set of all binary codewords $\mathbf{v} = (v_1,v_2,\dots,v_n)$ with length $n$ such that 
    \begin{itemize}
        \item $v_1  = v_2 = \dots = v_k = 1$, and $v_{k+1} = 0$, $v_{n-k} = 1$;
        \item $v_{n-k+1}, v_{n-k+2}, \dots, v_n = 0$;
        \item $(v_{k+1}, v_{k+2}, \dots, v_{n-k})$ does not contain $k$ consecutive $0$'s or $k$ consecutive $1$'s.
    \end{itemize}
    Define 
    \[
        \mathcal{V}_2^{(k)}(n)  = \bigcup_{i \ge 2k+2}^{n} V_2^{(k)}(i).
    \]
    Then $\mathcal{V}_2^{(k)}(n)$ is a binary variable-length non-overlapping code with maximum length $n$.
\end{construction_Bilotta}

\begin{example} Take $k=3$ and $n=10$, then Construction~II gives the code as
    \[
        \mathcal{V}_2^{(k)}(n) = \{11101000, 111011000, 111001000, 1110101000, 1110011000\}.    
    \]
\end{example}
By the generic method, we now generalize Construction~II by $\Phi_{I,J} (\mathcal{V}_2^{(k)}(n))$. 
It is readily seen that $\Phi_{I,J} (\mathcal{V}_2^{(k)}(n))$ is equivalent to the code by following construction.

\begin{construction_Phi_Bilotta}\label{cons: Phi_Bilotta}
    Let $I,J$ form a bipartition of $\mathbb{Z}_q$ with $q \ge 2$. 
    Let $n,k$ be integers such that $3 \le k \le \lfloor n/2 \rfloor -2$.
    Denote by $V_{I,J}^{(k)}(n)$ the set of all $q$-ary codewords $\mathbf{v} = (v_1,v_2,\dots,v_n)$ with length $n$ such that 
    \begin{itemize}
        \item $(v_1,v_2,\dots, v_k) \in J^k$, and $v_{k+1} \in I$, $v_{n-k} \in J$;
        \item $(v_{n-k+1}, v_{n-k+2}, \dots, v_n) \in I^k$;
        \item $(v_{k+1}, v_{k+2}, \dots, v_{n-k})$ is $(I^k \cup J^k)$-free.
    \end{itemize}
    In other words, $V_{I,J}^{(q)}(n) = J^k \times R_{I,J}^{(k)}(n - 2k) \times I^k$, where $R_{I,J}^{(k)}(l)$ is the set of $q$-ary $(I^k \cup J^k)$-free codewords of length $l$, that start with an element from $I$, and end with an element from $J$.
    Note that $R_{I,J}^{(k)}(0) = \emptyset$, $R_{I,J}^{(k)}(1) = \emptyset$, and $R_{I,J}^{(k)}(2) = I \times J$.
    Define 
    \[
    \mathcal{V}_{I,J}^{(k)}(n)  = \bigcup_{i \ge 2k+2}^{n} V_{I,J}^{(k)}(i).
    \]
    Then by Lemma~\ref{corollary: binary non-overlapping to q-ary non-overlapping}, $\mathcal{V}_{I,J}^{(k)}(n)$ is a $q$-ary variable-length non-overlapping code with maximum length $n$.
\end{construction_Phi_Bilotta}

\begin{remark}
    Construction~II can be viewed as a special case of Construction~II'.
    More precisely, taking $I = \{0\}$ and $J = \{1\}$ in Construction~II', we have $\mathcal{V}_{\{0\},\{1\}}^{(k)}(n) = \mathcal{V}_2^{(k)}(n)$.
\end{remark}

In the rest of this subsection, we again analyze the cardinality of codes by Construction~II' using a generating function approach.
The cardinality is related to the intermediate variable $r_{I,J}^{(k)}(l)$ defined as follows.
Let $k \ge 2$, define

\begin{equation}\label{equ: r_{I,J}}
    r_{I,J}^{(k)}(l) = \begin{cases}
        0 & \text{if } l < 0, \\
        1 & \text{if } l=0, \\
        |R_{I,J}^{(k)}(l)| & \text{if } l > 0, \text{ where $R_{I,J}^{(k)}(l)$ is defined in Construction~II'.}
    \end{cases}
\end{equation}

By Construction~II', for $k \ge 3$ and $n \ge 2k+2$, 

\begin{equation}\label{equ: size of variable-length code No.1}
    |\mathcal{V}_{I,J}^{(k)}(n)| = \sum_{i \ge 2k+2}^{n}  |V_{I,J}^{(k)}(i)|,
\end{equation}

and for $ i \ge 2k + 2$,

\begin{equation}\label{equ: size of subcode of variable-length code No.1}
    |V_{I,J}^{(k)}(i)| = |I|^k |J|^k r_{I,J}^{(k)}(i - 2k),
\end{equation}

It remains to determine the values of $r_{I,J}^{(k)}(\cdot)$.
For $I = \{0\}$ and $J = \{1\}$, it was solved by Barcucci, Bernini, Bilotta, and Pinzani~\cite{barcucci2017non} as follows. 
Firstly, they showed 
\begin{equation*} \label{equ: Barcucci recursive on r}
    r_{\{0\},\{1\}}^{(k)}(n) = \sum_{j=1}^{k}r_{\{0\},\{1\}}^{(k)}(n - j) - \sum_{j=k+1}^{2k-1}r_{\{0\},\{1\}}^{(k)}(n - j),
\end{equation*}
by giving a recurrence relation of the set $R_{\{0\},\{1\}}^{(k)}(n)$ in $2k-1$ terms of $R_{\{0\},\{1\}}^{(k)}(n-j)$ for $1 \le j \le 2k-1$. 
Then, they proved by induction that $r_{\{0\},\{1\}}^{(k)}(n)$ is equal to or differ by $1$ the sum of the $k-1$ preceding terms.
More precisely, for $k \ge 3$, $n \ge 2$, it was proved that 

\begin{equation}\label{equ: Barcucci induction on r}
    r_{\{0\},\{1\}}^{(k)}(n) = \sum_{j=1}^{k-1}r_{\{0\},\{1\}}^{(k)}(n-j) + d^{(k)}(n),
\end{equation}

where $d^{(k)}(n)$ is defined as
\[
    d^{(k)}(n) = \begin{cases}
        1 & \text{if } n \equiv 0 \pmod k , \\
        -1 & \text{if } n \equiv 1 \pmod k, \\
        0 &  \text{otherwise.}
    \end{cases}
\]

Finally, by the generating functions of $d^{(k)}(n)$ and $k$-generalized Fibonacci numbers (see~\cite[Proposition 3.2]{barcucci2017non}), the following was given. 

\begin{equation}\label{equ: Barcucci generating function on r}
    \sum_{n=0}^{\infty} r_{\{0\},\{1\}}^{(k)}(n)x^n  = \frac{1 - 2x + x^2}{1 - 2x + 2x^{k+1} - x^{2k}}.
\end{equation}
\begin{remark}
    In~\cite{barcucci2017non}, it is claimed that
    \[
        \sum_{n=0}^{\infty} r_{\{0\},\{1\}}^{(k)}(n)x^n = \frac{x^2 -2x^{k+1} + x^{2k}}{1 - 2x + 2x^{k+1} - x^{2k}} =  \frac{1 - 2x + x^2}{1 - 2x + 2x^{k+1} - x^{2k}} - 1.
    \]
    Taking $x=0$, it yields that $r_{\{0\},\{1\}}^{(k)}(0) = 0$, which should be $1$. 
    We correct this equation in Eq.~\eqref{equ: Barcucci generating function on r}.
\end{remark}

However, Eq.~\eqref{equ: Barcucci induction on r} cannot be easily generalized for other $I,J$.
Instead of trying to follow  Barcucci \textit{et al.}'s argument, we find a simple way to determine the generating function of $r_{I,J}^{(k)}(n)$ directly.

\begin{lemma}\label{lemma:  ogf of r_{I,J}}
    Let $I,J$ form a bipartition of $\mathbb{Z}_q$. For any integer $n, k,q$ larger than $1$, we have 
    \begin{equation}\label{equ: ogf of r_{I,J}}
        \sum_{n=0}^{\infty} r_{I,J}^{(k)}(n)x^n = \frac{|I||J|x^2 - qx + 1}{1 - qx + (|I|^k|J| + |I||J|^k)x^{k+1} - |I|^k|J|^k x^{2k}},
    \end{equation}
    Therefore, $r_{I,J}^{(k)}(n)$ is given as follows.
    \begin{equation}\label{equ: recursive realation on r_{I,J}}
        r_{I,J}^{(k)}(n) = 
        \begin{cases}
            0 & \text{if } n < 0 \text{ or } n = 1,\\
            1 & \text{if } n = 0, \\
            |I||J| & \text{if } n=2, \\
            q r_{I,J}^{(k)}(n-1) - (|I|^k|J| + |I||J|^k) r_{I,J}^{(k)}(n - k - 1) + |I|^k|J|^k r_{I,J}^{(k)}(n - 2k) & \text{if } n \ge 3.
        \end{cases}
    \end{equation}
\end{lemma}

\begin{proof}
    As in the proof of Lemma~\ref{lemma: ofg of J*J}, each codeword from $\bigcup_{n \ge 1}R_{I,J}^{(k)}(n)$ can be uniquely constructed by appending $q$-ary letters to the empty word as follows.

    \begin{itemize}
        \item[] Repeat the following two steps for $i \ge 0$ times.
        \begin{enumerate}
            \item[(a)] append less than $k$ consecutive letters from $I$;
            \item[(b)]  append less than $k$ consecutive letters from $J$;
        \end{enumerate}
    \end{itemize}

    Hence, we have
    \[
        \begin{aligned}
            \sum_{n=0}^{\infty} r_{I,J}^{(k)}(n)x^n & = \sum_{i=0}^{\infty} \left[ (|I|x + |I|^2x^2 + \dots + |I|^{k-1}x^{k-1})(|J|x + |J|^2x^2 + \dots + |J|^{k-1}x^{k-1}) \right]^i \\
            & = \sum_{i=0}^{\infty} \left( \frac{|I|x - |I|^kx^k}{1 - |I|x} \cdot \frac{|J|x - |J|^kx^k}{1 - |J|x} \right)^i \\
            & = \frac{1}{1 - \left( \frac{|I|x - |I|^kx^k}{1 - |I|x} \cdot \frac{|J|x - |J|^kx^k}{1 - |J|x} \right)} \\
            & = \frac{|I||J|x^2 - (|I| + |J|)x + 1}{1 - (|I| + |J|)x + (|I|^k|J| + |I||J|^k)x^{k+1} - |I|^k|J|^k x^{2k}}.
        \end{aligned}
    \]
Then Eq.~\eqref{equ: ogf of r_{I,J}} follows from the relation that $|I| + |J| = q$.

Note that Eq.~\eqref{equ: recursive realation on r_{I,J}} holds for $n \le 2$ by Eq.~\eqref{equ: r_{I,J}}. 
Multiplying the denominator to both sides of Eq.~\eqref{equ: ogf of r_{I,J}} and comparing the coefficient of $x^n$, we get the rest part of Eq.~\eqref{equ: recursive realation on r_{I,J}}.
\end{proof}

Setting $I = \{0\}$, $J = \{1\}$ in Eq.~\eqref{equ: ogf of r_{I,J}}, we reobtain Eq.~\eqref{equ: Barcucci generating function on r}.
Moreover, Eq.~\eqref{equ: recursive realation on r_{I,J}} gives a simpler recurrence relation on $r_{\{0\}, \{1\}}^{(k)}(n)$.

\begin{theorem}\label{thm: variable-length}
    Let $\mathcal{V}_{I,J}^{k}(n)$ be the code given by Construction~II' with $3 \le k \le \lfloor n/2 \rfloor -2$.
    We have
    \begin{equation}\label{equ: ogf of variable-length code No.1}
        \sum_{n=2k+2}^{\infty} |\mathcal{V}_{I,J}^{(k)}(n)|x^n = \frac{|I|^k |J|^k x^{2k}}{1-x} \left( \frac{|I||J|x^2 - qx + 1}{1 - qx + (|I|^k|J| + |I||J|^k)x^{k+1} - |I|^k|J|^k x^{2k}} -1 \right).
    \end{equation}
\end{theorem}

\begin{proof}
By Eq.s~\eqref{equ: size of subcode of variable-length code No.1} and~\eqref{equ: ogf of r_{I,J}}, we have 

\begin{equation*} \label{equ: ofg of subcode of variable-length code No.1}
    \begin{aligned}
        \sum_{i=2k + 2}^{\infty} |V_{I,J}^{(k)}(i)|x^i & = \sum_{i=2k+2}^{\infty}|I|^k |J|^k r_{I,J}^{(k)}(i-2k)x^i \\
                & = |I|^k |J|^k x^{2k} \sum_{n=2}^{\infty} r_{I,J}^{(k)}(n)x^n \\
                & = |I|^k |J|^k x^{2k} \left( \frac{|I||J|x^2 - qx + 1}{1 - qx + (|I|^k|J| + |I||J|^k)x^{k+1} - |I|^k|J|^k x^{2k}} -1 \right).
    \end{aligned}   
\end{equation*}

The desired result then follows from Eq.~\eqref{equ: size of variable-length code No.1}.
\end{proof}

In particular, by Theorem~\ref{thm: variable-length}, we have
\[
    \sum_{n=2k+2}^{\infty} |\mathcal{V}_{\{0\},\{1\}}^{(k)}(n)|x^n = \frac{x^{2k}(x - x^k)^2}{(1-x)(1-x^k)(1-2x+x^k)},
\]
which coincides with the generating function given in~\cite[Eq.~(9)]{bilotta2017variable}.
Moreover, Bilotta~\cite{bilotta2017variable} showed that
\[
    \limsup_{n \to \infty} |\mathcal{V}_{\{0\}, \{1\}}^{(k)}(n)|^{1/n} = 2(1-\epsilon_{k-1}) ,
\]
where $\epsilon_{k-1}$ is given in Eq.~\eqref{equ: epsilon_k}.
We generalize this result as follows.

\begin{theorem}
    Suppose that $n,k,q$ are integers larger than $1$ such that $8 \le 2k+2 \le n$. Let $I,J$ form a bipartition of $\mathbb{Z}_q$ such that $|I| = |J|$, then we have
    \begin{equation}\label{equ: asymptotic size of variable-length code No.1}
        \limsup_{n \to \infty} |\mathcal{V}_{I,J}^{(k)}(n)|^{1/n} =  q(1-\epsilon_{k-1}),
    \end{equation}
    where $\epsilon_{k-1}$ is given in Eq.~\eqref{equ: epsilon_k}.
\end{theorem}

\begin{proof}
    The proof idea is similar to that of Theorem~\ref{thm: upper bound of ConstructionI}.
    Let $q = 2r$.
    By Eq.~\eqref{equ: ogf of variable-length code No.1}, we have
\[  
    \begin{aligned}
        \sum_{n=2k+2}^{\infty} |\mathcal{V}_{I,J}^{(k)}(n)|x^n & = \frac{(rx)^{2k}(rx - (rx)^k)^2}{(1-x)(1-(rx)^k)(1-2rx+(rx)^k)} \\
        & = \frac{(rx)^{2k}(rx - (rx)^k)^2}{(1-x)(1-(rx)^k)(1-rx)(1-\sum_{i=0}^{k-1}(rx)^i)}.
    \end{aligned}    
\]

By Pringsheim's Theorem, the covering radius $R$ of the series $\sum_{n=2k+2}^{\infty} |\mathcal{V}_{\{0\},\{1\}}^{(k)}(n)|x^n$ is the smallest positive real root of the denominator. By considering the reciprocal polynomial of $1-\sum_{i=0}^{k-1}(rx)^i$ and Lemma~\ref{lemma: root of denominator}, the unique real root of $1-\sum_{i=0}^{k-1}(rx)^i = 0$ is $\left[ 2p(1-\epsilon_{k-1})\right]^{-1} = \left[ q(1-\epsilon_{k-1})\right]^{-1}$, where $\epsilon_{k-1}$ is given in Eq.~\eqref{equ: epsilon_k}. Therefore $R = \left[ q(1-\epsilon_{k-1})  \right]^{-1}$ and 

\[
    \limsup_{n \to \infty} |\mathcal{V}_{I,J}^{(k)}(n)|^{1/n} = R^{-1} = q(1-\epsilon_{k-1}).
\]
The proof is then completed.
\end{proof}

For $|I| = |J|$ and all fixed $\epsilon >0$, there exists an integer $N$ such that $|\mathcal{V}_{I,J}^{(k)}(n)| < \left[  q(1-\epsilon_{k-1}) + \epsilon \right]^n$ holds for all $n>N$, and $|\mathcal{V}_{I,J}^{(k)}(n)| > \left[  q(1-\epsilon_{k-1}) -\epsilon \right]^n$ holds for infinitely many values of $n$.

Note that our generalization also works for other constructions for binary non-overlapping codes. 
Numerical results show that our generalization indeed produces large $q$-ary non-overlapping codes, especially when $q$ is large.
Details are presented in Appendix.

\section{Recursive upper bounds on variable-length non-overlapping codes} \label{sec: upper bounds on variable-length non-overlapping codes}

Throughout this section, denote by $\mathcal{J}_q(n) = \bigcup_{i = h}^n J_q(i)$ a $q$-ary non-overlapping code, and suppose that each subset $J_q(i)$ contains all the codewords in $\mathcal{J}_q(n)$ with length $i$. 
In particular, $\mathcal{J}_q(n)$ is a fixed-length non-overlapping code if $h=n$. 
Let $B_{\mathcal{J}_q(n)}(m)$ be the set of  $q$-ary words with length $m$ that do not contain codewords from $\mathcal{J}_q(n)$ as subwords. 
Define 
\begin{equation*}
    b_{\mathcal{J}_q(n)}(m) = 
    \begin{cases}
        0, & \text{if $m < 0$} \\
        1, & \text{if $m = 0$}  \\
        |B_{\mathcal{J}_q(n)}(m)|, & \text{if $m > 0$}
    \end{cases}
\end{equation*}
Recall that Bilotta~\cite{bilotta2017variable} showed a recursive upper bound on $J_2(n)$ that

\begin{equation*}
    |J_2(n)| < \frac{2^n}{n+1} - \sum_{i=1}^{h-2}2^i |J_2(n-i)| - \frac{1}{2}\sum_{i=h}^{n+1-h} b_{\mathcal{J}_2(n)}(i) \cdot |J_2(n+1-i)|,
\end{equation*}
where the exact expression of $b_{\mathcal{J}_2(n)}(i)$ is left open. 
In the following, we first give an exact expression of $b_{\mathcal{J}_q(n)}(m)$ with $q \ge 2$, and then generalize this recursive bound. 

\begin{theorem}
    Let notations be the same as above. We have
    \begin{equation}\label{equ: exact expression of B_m}
        \begin{aligned}
            b_{\mathcal{J}_q(n)}(m) & = 
            & \sum (-1)^{r} \frac{(t_1 + r)!}{t_1! t_h!\cdots t_n!} q^{t_1}|J_q(h)|^{t_h}|J_q(h+1)|^{t_{h+1}} \cdots |J_q(n)|^{t_n},
        \end{aligned}
    \end{equation}
    where $r = t_h + t_{h+1} + \dots + t_n$
    and the summation is over all nonnegative integers $t_1,t_h,\dots, t_n$ such that
    \[
        {t_1+ht_h + (h+1)t_{h+1} + \dots + nt_n = m \ge 0}.
    \]
\end{theorem}

\begin{proof}
Let $m \ge h$. Let $P = B_{\mathcal{J}_q(n)}(m -1) \times \mathbb{Z}_q$ be the set of $q$-ary words of length $m$ that begin with a codeword from $B_{\mathcal{J}_q(n)}(m -1)$ and are followed by an arbitrary $q$-ary letter. 
Clearly, $B_{\mathcal{J}_q(n)}(m) \subseteq P$.
Let $T_i = B_{\mathcal{J}_q(n)}(m -i) \times J_q(i)$ for $h \le i\le n$. 
Thus $T_i$'s are pairwise disjoint and $P \setminus B_{\mathcal{J}_q(n)}(m) = \bigcup_{i=h}^n T_i$ since $\mathcal{J}_q(n)$ is non-overlapping. 
Note that $|P| = q \cdot b_{\mathcal{J}_{q}(n)}(m-1)$, and $|T_i| =  b_{\mathcal{J}_q(n)}(m-i) \cdot |J_q(i)|$ for $h \le i\le n$.
Hence, for any $m \ge h$, we have
\begin{equation*}\label{equ: recurrence relation on b_m}
    q \cdot b_{\mathcal{J}_{q}(n)}(m-1) - b_{\mathcal{J}_{q}(n)}(m) = \sum_{i=h}^{n} b_{\mathcal{J}_q(n)}(m-i) \cdot |J_q(i)| ,
\end{equation*}
and $b_{\mathcal{J}_q(n)}(m)=q^m$ for $0 \le m < h$. Therefore 
\begin{equation}\label{equ: ogf of B_m}
    \begin{aligned}
        \sum_{m=0}^{\infty} b_{\mathcal{J}_q(n)}(m)x^m & = \frac{1}{1 - qx + \sum_{i=h}^n |J_q(i)| x^{i}} \\
                                                        & = 1 + \left( qx - \sum_{i=h}^n |J_q(i)| x^{i}\right) + \left( qx - \sum_{i=h}^n |J_q(i)| x^{i}\right)^2 + \dots .
    \end{aligned}
\end{equation}
Let $b_k$ be the coefficient of $x^m$ in $\left( qx - \sum_{i=h}^n |J_q(i)| x^{i}\right)^k$ for $k \ge 0$. 
Then we have 
\begin{equation*}
    b_k = \sum (-1)^{k - t_1} \frac{k!}{t_1! t_h!\cdots t_n!} q^{t_1}|J_q(h)|^{t_h}|J_q(h+1)|^{t_{h+1}} \cdots |J_q(n)|^{t_n},
\end{equation*}
where $t_1 + t_h + \dots + t_n = k$ and the summation is over all nonnegative integers $t_1,t_h,\dots, t_n$ such that 
\[
    t_1+ht_h + (h+1)t_{h+1} + \dots + nt_n = m \ge 0.
\]
It implies $k \le m$, and Eq.~\eqref{equ: exact expression of B_m} then follows from the relation $b_{\mathcal{J}_q(n)}(m) = \sum_{k=0}^{m}b_k$.
\end{proof}

We remark that Eq.~\eqref{equ: exact expression of B_m} has a combinatorial interpretation.
Denote by $E_{j,i}$ be the set of $m$-length $q$-ary words $(s_1, s_2, \dots, s_m)$ that contain a subword $(s_j, \dots, s_{j + i - 1}) \in J_q(i)$. 
By the principle of inclusion and exclusion, Eq.~\eqref{equ: exact expression of B_m} counts the number of elements in $\mathbb{Z}_q^m \setminus \bigcup E_{j,i} = B_{\mathcal{J}_{q}(n)}(m)$ for all $j + i - 1 \le n $ and $ h \le i \le n$.

Given $|J_q(i)|$ for $h \le i < n$, we are able to derive an upper bound on $|J_q(n)|$ as follows.

\begin{theorem}\label{thm: our recursive bound}
    Let notations be the same as above and $1 \le m < h$. We have
    \begin{equation}\label{equ: our recursive bound}
        |J_q(n)| < \frac{q^{n}}{m+n} - \frac{1}{q^m} \sum_{i=h}^{n-1}b_{\mathcal{J}_q(n)}(m + n -i) \cdot |J_q(i)|,
    \end{equation}
    where $b_{\mathcal{J}_q(n)}(m + n - i)$ is given in Eq.~\eqref{equ: exact expression of B_m}.
\end{theorem}

\begin{proof}
    We follow the argument by Bilotta~\cite[Proposition 7]{bilotta2017variable}.
    Let $X \subseteq \mathbb{Z}_q^{m + n}$ be the set of all $q$-ary words with length $(m + n)$ that each contains exactly one codeword from $\mathcal{J}_q(n)$ as a cyclic subword. 
    For every $\mathbf{\omega} \in X$, there are $(m + n)$ possible positions for the cyclic subword from $\mathcal{J}_q(n)$. Since codewords in $\mathcal{J}_q(n)$ do not overlap with each other, we have
    
    \[
        \begin{aligned}
            |X| & = \sum_{i=h}^{n} (m + n) \cdot b_{\mathcal{J}_q(n)}(m + n - i) \cdot |J_q(i)| \\
                & = (m + n)q^m|J_q(n)| + \sum_{i=h}^{n-1} (m + n) \cdot b_{\mathcal{J}_q(n)}(m + n - i) \cdot |J_q(i)|.
        \end{aligned}
    \]

    On the other hand, $|X| \le q^{m + n} - q < q^{m + n}$ since the $q$ constant words of length $m + n$ (e.g., $00\dots 0$) cannot have a codeword from $\mathcal{J}_q(n)$ as a cyclic subword. 
    We have
    \[
        |X| = (m + n)q^m|J_q(n)| + \sum_{i=h}^{n-1} (m + n) \cdot b_{\mathcal{J}_q(n)}(m + n - i) \cdot |J_q(i)| < q^{m + n},
    \]
    and Eq.~\eqref{equ: our recursive bound} then follows.
    Note that $b_{\mathcal{J}_q(n)}(m + n - i)$ does not involve $|J_q(n)|$.
\end{proof}

We remark that Eq.~\eqref{equ: our recursive bound} is a generalization of Eq.s~\eqref{equ: bilotta's upper bound} and~\eqref{equ: com bound}.
Take $m = 1$ in Eq.~\eqref{equ: our recursive bound}, we have
\begin{equation*}
    \begin{aligned}
        |J_q(n)| & < \frac{q^n}{n + 1} - \frac{1}{q}\sum_{i = h}^{n-1}b_{\mathcal{J}_q}(n + 1 - i) \cdot |J_q(i)| \\
                    & = \frac{q^n}{n + 1} - \frac{1}{q}\sum_{i = h}^{n - h + 1}b_{\mathcal{J}_q}(n + 1 - i) \cdot |J_q(i)| - \frac{1}{q}\sum_{i = n - h + 2}^{n - 1}q^{n + 1 - i}  |J_q(i)| \\
                    & = \frac{q^n}{n + 1} - \sum_{i = 1}^{h-2} q^{i}|J_q(n - i)| - \frac{1}{q}\sum_{i = h}^{n - h + 1}b_{\mathcal{J}_q}(n + 1 - i) \cdot |J_q(i)|.
    \end{aligned}
\end{equation*}        
On the other hand, take $h = n$. Eq.~\eqref{equ: our recursive bound} becomes $q^n/(m + n)$ since the summation of the RHS vanishes.

Recall that $C(n,q)$ is the maximum size of $q$-ary non-overlapping codes of length $n$. 
Clearly, $|J_q(n)| \le C(n,q)$ for $J_q(n)$ is also non-overlapping. 
Also, by Theorem~\ref{thm: our recursive bound}, we have
\begin{equation}\label{equ: our recursive bound min}
    |J_q(n)| < \min_{1 \le m < h} \left\{   \frac{q^{n}}{m+n} - \frac{1}{q^m} \sum_{i=h}^{n-1}b_{\mathcal{J}_q(n)}(m + n -i) \cdot |J_q(i)| \right \}.
\end{equation}
In particular, for $h = n$,
\begin{equation*}
    |J_q(n)| < \min_{1 \le m < n} \left\{ \frac{q^{n}}{m+n} \right \} = \frac{q^n}{2n-1}.
\end{equation*}

Ideally, we hope Eq.~\eqref{equ: our recursive bound min} provides a tighter upper bound on $|J_q(n)|$ than Eq.s~\eqref{equ: lev bound} and~\eqref{equ: com bound}.
However, Eq.~\eqref{equ: our recursive bound min} depends on the given non-overlapping code $\mathcal{J}_q(n)$ and cannot be easily analysed.
In fact, it seems hard to find a tight bound of $|J_q(n)|$ only by $|J_q(i)|$ for $h \le i < n$, and $|J_q(n)|$ is highly related to the structure of $J_q(i)$.
For example, assume that $0001 \in \mathcal{J}_2(5)$. 
It can be readily seen that $\mathcal{J}_2(5) = \{0001\}$ and $|J_q(5)| = 0$.

A more challenging problem is to find a direct upper bound for the maximum size of $\mathcal{J}_q(n)$.
Trivially, $|\mathcal{J}_q(n)| = \sum_{i = h}^{n}J_q(i) \le \sum_{i = h}^{n} C(i,q)$ and we are interested in finding a tighter bound.
To see the difficulty of this problem, we remark that Eq.~\eqref{equ: lev bound} cannot be easily generalized to bound $|\mathcal{J}_q(n)|$.
Let $f(x) = 1 - qx + \sum_{i=h}^n |J_q(i)| x^{i}$ be the denominator of Eq.~\eqref{equ: ogf of B_m}.
By Pringsheim's Theorem and Descartes' rule of signs, $f(x)$ must have 2 positive real roots.
Note that $f(0) > 0$ and $f(x)$ has a unique minima at $x = x_0$ if $h = n$, where 
\[
x_0 = \left( \frac{q}{n |J_q(n)|} \right)^{1/(n-1)}.
\]
Hence Eq.~\eqref{equ: lev bound} follows from $f(x_0) \le 0$.
However, when $h < n$, the minima $x_0$ of $f(x)$ cannot be explicitly given, and $f(x_0) \le 0$ may not yield an upper bound of $\sum_{i=h}^n|J_q(i)|$.

\section{Conclusions} \label{sec: conclusions}

We give a generic method to extend binary non-overlapping codes to $q$-ary, and investigate this generalization on Construction~I ( to Construction~I') and Construction~II (to Construction~II').
Construction~I' provides a large non-expandable fixed-length non-overlapping codes whose sizes can be explicitly counted, and Construction~II' is the first construction for $q$-ary variable-length non-overlapping codes.
By the generating function approach, we establish new results on their cardinalities and greatly simplify some previous arguments.
In this process, we also find a new combinatorial interpretation and non-expandability for Construction~I'.
In addition, we answer an open problem by Bilotta~\cite{bilotta2017variable} and give a recursive upper bound on the maximum size of $q$-ary variable-length non-overlapping codes.

Further studies on non-overlapping codes may be devoted to more constructions and tighter bounds on the cardinality of codes.
In particular, it would be interesting but difficult to find a nontrivial direct upper bound on the maximum size of variable-length non-overlapping codes.
Moreover, our generalization on other constructions for non-overlapping codes could be further studied.


\begin{thebibliography}{10}
    \providecommand{\url}[1]{#1}
    \csname url@samestyle\endcsname
    \providecommand{\newblock}{\relax}
    \providecommand{\bibinfo}[2]{#2}
    \providecommand{\BIBentrySTDinterwordspacing}{\spaceskip=0pt\relax}
    \providecommand{\BIBentryALTinterwordstretchfactor}{4}
    \providecommand{\BIBentryALTinterwordspacing}{\spaceskip=\fontdimen2\font plus
    \BIBentryALTinterwordstretchfactor\fontdimen3\font minus
      \fontdimen4\font\relax}
    \providecommand{\BIBforeignlanguage}[2]{{%
    \expandafter\ifx\csname l@#1\endcsname\relax
    \typeout{** WARNING: IEEEtran.bst: No hyphenation pattern has been}%
    \typeout{** loaded for the language `#1'. Using the pattern for}%
    \typeout{** the default language instead.}%
    \else
    \language=\csname l@#1\endcsname
    \fi
    #2}}
    \providecommand{\BIBdecl}{\relax}
    \BIBdecl
    
    \bibitem{levenshtein1964decoding}
    V.~I. Levenshtein, ``Decoding automata which are
      invariant with respect to their initial state,'' (in Russian),
      \emph{Problems Cybern.}, vol.~12, pp. 125--136,
      1964.
    
    \bibitem{levenshtein1970maximum}
    V.~I. Levenshtein, ``Maximum number of words in codes without overlaps,'' \emph{Problems
      Inf. Transmiss.}, vol.~6, no.~4, pp. 355--357, 1970.
    
    \bibitem{levenshtein2004combinatorial}
    V.~I. Levenshtein, ``Combinatorial problems motivated by comma-free codes,'' \emph{J.
      Combinat. Designs}, vol.~12, no.~3, pp. 184--196, 2004.
    
    \bibitem{gilbert1960synchronization}
    E.~Gilbert, ``Synchronization of binary messages,'' \emph{IRE Trans. on Inf.
      Theory}, vol.~6, no.~4, pp. 470--477, 1960.
    
    \bibitem{chee2013cross}
    Y.~M. Chee, H.~M. Kiah, P.~Purkayastha, and C.~Wang, ``Cross-bifix-free codes
      within a constant factor of optimality,'' \emph{{IEEE} Trans. Inf. Theory},
      vol.~59, no.~7, pp. 4668--4674, 2013.
    
    \bibitem{van2000frame}
    A.~J. de~Lind~van Wijngaarden and T.~J. Willink, ``Frame synchronization using
      distributed sequences,'' \emph{{IEEE} Trans. Commun.}, vol.~48, no.~12, pp.
      2127--2138, 2000.
    
    \bibitem{bajic2004distributed}
    D.~Bajic and J.~Stojanovic, ``Distributed sequences and search process,'' in
      \emph{Proc. {IEEE} Int. Conf. Commun., {ICC} 2004, Paris, France, 20-24 June
      2004}.\hskip 1em plus 0.5em minus 0.4em\relax {IEEE}, 2004, pp. 514--518.
    
    \bibitem{bilotta2012new}
    S.~Bilotta, E.~Pergola, and R.~Pinzani, ``A new approach to cross-bifix-free
      sets,'' \emph{{IEEE} Trans. Inf. Theory}, vol.~58, no.~6, pp. 4058--4063,
      2012.
    
    \bibitem{blackburn2015non}
    S.~R. Blackburn, ``Non-overlapping codes,'' \emph{{IEEE} Trans. Inf. Theory},
      vol.~61, no.~9, pp. 4890--4894, 2015.
    
    \bibitem{barcucci2016cross}
    E.~Barcucci, S.~Bilotta, E.~Pergola, R.~Pinzani, and J.~Succi,
      ``Cross-bifix-free sets via Motzkin paths generation,'' \emph{RAIRO - Theor.
      Inform. and Appl.}, vol.~50, no.~1, pp. 81--91, 2016.
    
    \bibitem{levy2018mutually}
    M.~Levy and E.~Yaakobi, ``Mutually uncorrelated codes for {DNA} storage,''
      \emph{{IEEE} Trans. Inf. Theory}, vol.~65, no.~6, pp. 3671--3691, 2018.
    
    \bibitem{yazdi2015rewritable}
    S.~M. H.~T. Yazdi, Y.~Yuan, J.~Ma, H.~Zhao, and O.~Milenkovic, ``A rewritable,
      random-access {DNA}-based storage system,'' \emph{Scientific Reports},
      vol.~5, p. 14138, 2015.
    
    \bibitem{crochemore2007algorithms}
    M.~Crochemore, C.~Hancart, and T.~Lecroq, \emph{Algorithms on Strings}.\hskip
      1em plus 0.5em minus 0.4em\relax Cambridge University Press, 2007.
    
    \bibitem{berstel2010codes}
    J.~Berstel, D.~Perrin, and C.~Reutenauer, \emph{Codes and Automata}, ser.
      Encyclopedia of mathematics and its applications.\hskip 1em plus 0.5em minus
      0.4em\relax Cambridge University Press, 2010, vol. 129.
    
    \bibitem{bajic2014simple}
    D.~Bajic and T.~Loncar-Turukalo, ``A simple suboptimal construction of
      cross-bifix-free codes,'' \emph{Cryptogr. Commun.}, vol.~6, no.~1, pp.
      27--37, 2014.
    
    \bibitem{bernini2014prefix}
    A.~Bernini, S.~Bilotta, R.~Pinzani, A.~Sabri, and V.~Vajnovszki, ``Prefix
      partitioned {Gray} codes for particular cross-bifix-free sets,''
      \emph{Cryptogr. Commun.}, vol.~6, no.~4, pp. 359--369, 2014.
    
    \bibitem{bernini2017gray}
    A.~Bernini, S.~Bilotta, R.~Pinzani, and V.~Vajnovszki, ``A {Gray} code for
      cross-bifix-free sets,'' \emph{Math. Struct. in Comput. Sci.}, vol.~27,
      no.~2, pp. 184--196, 2017.
    
    \bibitem{bilotta2017variable}
    S.~Bilotta, ``Variable-length non-overlapping codes,'' \emph{{IEEE} Trans. Inf.
      Theory}, vol.~63, no.~10, pp. 6530--6537, 2017.
    
    \bibitem{wolfram1996solving}
    D.~A. Wolfram, ``Solving generalized {F}ibonacci recurrences,''
      \emph{{Fibonacci} Quart.}, vol.~36, no.~2, pp. 129--145, 1998.
    
    \bibitem{stein2010complex}
    E.~M. Stein and R.~Shakarchi, \emph{Complex Analysis}, ser. Princeton Lectures
      in Analysis.\hskip 1em plus 0.5em minus 0.4em\relax Princeton University
      Press, Princeton, NJ, 2003, vol.~2.
    
    \bibitem{flajolet2009analytic}
    \BIBentryALTinterwordspacing
    P.~Flajolet and R.~Sedgewick, \emph{Analytic Combinatorics}.\hskip 1em plus
      0.5em minus 0.4em\relax Cambridge University Press, Cambridge, 2009.
    \BIBentrySTDinterwordspacing
    
    \bibitem{barcucci2017non}
    E.~Barcucci, A.~Bernini, S.~Bilotta, and R.~Pinzani, ``Non-overlapping
      matrices,'' \emph{Theoret. Comput. Sci.}, vol. 658, pp. 36--45, 2017.
    
    \end{thebibliography}


\appendix

\subsection{Fixed-length non-overlapping codes}

Tables~\ref{tab: 3-ary fixed-length non-overlapping codes} -~\ref{tab: 6-ary fixed-length non-overlapping codes} list the largest cardinalities of $q$-ary ($3 \le q \le 6$) fixed-length non-overlapping codes by
\begin{enumerate}
    \item Construction~I;
    \item Construction~I';
    \item Bilotta \textit{et al.}'s construction~\cite{bilotta2012new}; this is a binary construction based on Dyck paths. We generalize it to $q$-ary by the method in Definition~\ref{def: binary to q-ary codes}.
    \item Barcucci \textit{et al}'s construction~\cite{barcucci2016cross}; this is a $q$-ary ($q \ge 3$) construction based on colored Motzkin paths.
\end{enumerate}
For each construction, the largest cardinality is obtained by properly choosing parameters.
For the cardinalities of binary fixed-length non-overlapping codes, we refer to~\cite[Table 2]{bajic2014simple}.

\subsection{Variable-length non-overlapping codes}
Tables~\ref{tab: 3-ary V_I,J} and~\ref{tab: 4-ary V_I,J} list the largest cardinalities of $3$-ary and $4$-ary $\mathcal{V}_{I,J}^{(k)}(n)$ for $8 \le n \le 23$ by properly choosing $I,J$. 
For the binary case, we refer to~\cite[Table II]{bilotta2017variable}.


\begin{table}[h!]
    \centering
    \caption{Cardinalities of $3$-ary fixed-length non-overlapping codes}
    \label{tab: 3-ary fixed-length non-overlapping codes}
    \begin{tabular}{c|llll}
    \toprule
    \hline
    $n$     & Construction~I & Construction~I' & Bilotta \textit{et al.}~\cite{bilotta2012new} & Barcucci \textit{et al.}~\cite{ barcucci2016cross}  \\ 
    \hline
    3  & 4      & 4      & 4      & 4      \\
    4  & 8      & 8      & 4      & 7      \\
    5  & 16     & 16     & 16     & 16     \\
    6  & 32     & 32     & 24     & 36     \\
    7  & 88     & 88     & 80     & 87     \\
    8  & 240    & 240    & 128    & 210    \\
    9  & 656    & 656    & 448    & 535    \\
    10 & 1792   & 1792   & 736    & 1350   \\
    11 & 4896   & 4896   & 2688   & 3545   \\
    12 & 13376  & 13376  & 4608   & 9205   \\
    13 & 36544  & 36544  & 16896  & 24698  \\
    14 & 99840  & 99840  & 29056  & 65467  \\
    15 & 272768 & 272768 & 109824 & 178375 \\
    16 & 745216 & 745216 & 194560 & 480197 \\ 
    \hline
    \bottomrule
    \end{tabular}
\end{table}

\begin{table}[h!]
    \centering
    \caption{Cardinalities of $4$-ary fixed-length non-overlapping codes}
    \label{tab: 4-ary fixed-length non-overlapping codes}
    \begin{tabular}{c|llll}
    \toprule
    \hline
    $n$     & Construction~I & Construction~I' & Bilotta \textit{et al.}~\cite{bilotta2012new} & Barcucci \textit{et al.}~\cite{ barcucci2016cross}  \\ 
    \hline
    3  & 9        & 9        & 9        & 9        \\
    4  & 27       & 27       & 16       & 25       \\
    5  & 81       & 81       & 64       & 72       \\
    6  & 243      & 243      & 192      & 223      \\
    7  & 729      & 729      & 640      & 712      \\
    8  & 2187     & 2187     & 2048     & 2334     \\
    9  & 7371     & 7371     & 7168     & 7868     \\
    10 & 27945    & 27945    & 23552    & 26731    \\
    11 & 105948   & 105948   & 86016    & 93175    \\
    12 & 401679   & 401679   & 294912   & 324520   \\
    13 & 1522881  & 1522881  & 1081344  & 1157031  \\
    14 & 5773680  & 5773680  & 3719168  & 4104449  \\
    15 & 21889683 & 21889683 & 14057472 & 14874100 \\
    16 & 82990089 & 82990089 & 49807360 & 53514974 \\
    \hline
    \bottomrule
    \end{tabular}
\end{table}

\begin{table}[h!]
    \centering
    \caption{Cardinalities of $5$-ary fixed-length non-overlapping codes}
    \label{tab: 5-ary fixed-length non-overlapping codes}
    \begin{tabular}{c|llll}
    \toprule
    \hline
    $n$     & Construction~I & Construction~I' & Bilotta \textit{et al.}~\cite{bilotta2012new} & Barcucci \textit{et al.}~\cite{ barcucci2016cross}  \\ 
    \hline
    3  & 16         & 18         & 18         & 16         \\
    4  & 64         & 64         & 36         & 61         \\
    5  & 256        & 256        & 216        & 224        \\
    6  & 1024       & 1024       & 648        & 900        \\
    7  & 4096       & 4096       & 3240       & 3595       \\
    8  & 16384      & 16384      & 10368      & 15014      \\
    9  & 65536      & 65536      & 54432      & 63135      \\
    10 & 262144     & 278964     & 178848     & 271136     \\
    11 & 1048576    & 1219860    & 979776     & 1178677    \\
    12 & 4870144    & 5333364    & 3359232    & 5167953    \\
    13 & 23515136   & 23515136   & 18475776   & 22986100   \\
    14 & 113541120  & 113541120  & 63545472   & 102403229  \\
    15 & 548225024  & 548225024  & 360277632  & 463098075  \\
    16 & 2647064576 & 2647064576 & 1276508160 & 2089302415 \\
    \hline
    \bottomrule
    \end{tabular}
\end{table}

\begin{table}[h!]
    \centering
    \caption{Cardinalities of $6$-ary fixed-length non-overlapping codes}
    \label{tab: 6-ary fixed-length non-overlapping codes}
    \begin{tabular}{c|llll}
    \toprule
    \hline
    $n$     & Construction~I & Construction~I' & Bilotta \textit{et al.}~\cite{bilotta2012new} & Barcucci \textit{et al.}~\cite{ barcucci2016cross}  \\ 
    \hline
    3  & 25          & 32          & 32          & 25          \\
    4  & 125         & 128         & 81          & 121         \\
    5  & 625         & 625         & 512         & 550         \\
    6  & 3125        & 3125        & 2187        & 2739        \\
    7  & 15625       & 15625       & 10935       & 13260       \\
    8  & 78125       & 78125       & 52488       & 67740       \\
    9  & 390625      & 390625      & 275562      & 342676      \\
    10 & 1953125     & 1953125     & 1358127     & 1787415     \\
    11 & 9765625     & 10027008    & 7440174     & 9324647     \\
    12 & 48828125    & 54788096    & 38263752    & 49456240    \\
    13 & 244140625   & 299368448   & 210450636   & 263776127   \\
    14 & 1220703125  & 1635778560  & 1085733963  & 1417981855  \\
    15 & 7068828125  & 8938061824  & 6155681103  & 7688015908  \\
    16 & 41381640625 & 48838475776 & 32715507960 & 41785951916 \\
    \hline
    \bottomrule
    \end{tabular}
\end{table}

\begin{table}[h!]
    \centering
    \caption{Cardinalities of $3$-ary $\mathcal{V}_{I,J}^{(k)}(n)$}
    \label{tab: 3-ary V_I,J}
    \begin{tabular}{c|llllllll}
    \toprule
    \hline
    \diagbox{$n$}{$k$} & 3 & 4 & 5 & 6 & 7 & 8 & 9 & 10  \\
    \hline
    8 & 16 \\
    9 & 64 \\
    10 & 128 & 32 \\
    11 & 320 & 128 \\
    12 & 800 & 416 & 64 \\
    13 & 1760 & 992 & 256 \\
    14 & 4128 & 2720 & 832 & 128 \\
    15 & 9696 & 7328 & 2560 & 512 \\
    16 & 22112 & 19680 & 6656 & 1664 & 256 \\
    17 & 51296 & 51552 & 18944 & 5120 & 1024 \\
    18 & 119008 & 137312 & 53632 & 15488 & 3328 & 512 \\
    19 & 274144 & 365024 & 151168 & 42368 & 10240 & 2048 \\
    20 & 634336 & 969824 & 425216 & 123008 & 30976 & 6656 & 1024 \\
    21 & 1467616 & 2571104 & 1188608 & 356480 & 93184 & 20480 & 4096 \\
    22 & 3389664 & 6828896 & 3341568 & 1031552 & 263168 & 61952 & 13312 & 2048 \\
    23 & 7837920 & 18132320 & 9388800 & 2980736 & 773120 & 186368 & 40960 & 8192 \\
    \hline
    \bottomrule
    \end{tabular}
\end{table}

\begin{table}[h!]
    \centering
    \caption{Cardinalities of $4$-ary $\mathcal{V}_{I,J}^{(k)}(n)$}
    \label{tab: 4-ary V_I,J}
    \begin{tabularx}{1.02\textwidth}{c|llllllll}
    \toprule
    \hline
    \diagbox{$n$}{$k$} & 3 & 4 & 5 & 6 & 7 & 8 & 9 & 10  \\
    \hline
    8 & 256 \\
    9 & 1280 \\
    10 & 3328 & 1024 \\
    11 & 11520 & 5120 \\
    12 & 40192 & 21504 & 4096 \\
    13 & 122112 & 70656 & 20480 \\
    14 & 400640 & 267264 & 86016 & 16384 \\
    15 & 1318144 & 988160 & 348160 & 81920 \\
    16 & 4201728 & 3675136 & 1265664 & 344064 & 65536 \\
    17 & 13638912 & 13374464 & 4935680 & 1392640 & 327680 \\
    18 & 44309760 & 49288192 & 19091456 & 5586944 & 1376256 & 262144 \\
    19 & 142875904 & 181408768 & 73617408 & 21315584 & 5570560 & 1310720 \\
    20 & 462691584 & 667948032 & 284381184 & 84230144 & 22347776 & 5505024 & 1048576 \\
    21 & 1498684672 & 2454721536 & 1093881856 & 331694080 & 89456640 & 22282240 & 5242880 \\
    22 & 4845739264 & 9031390208 & 4218638336 & 1304772608 & 349503488 & 89391104 & 22020096 & 4194304 \\
    23 & 15683820800 & 33224135680 & 16264679424 & 5129977856 & 1389690880 & 357826560 & 89128960 & 20971520 \\
    \hline
    \bottomrule
    \end{tabularx}
\end{table}

\end{document}